\documentclass[letterpaper, 10 pt, conference]{ieeeconf}  

\IEEEoverridecommandlockouts                              

\usepackage{graphicx} 
\usepackage{amsmath} 
\usepackage{amssymb}  

\usepackage{amsthm}
\usepackage{cite}
\usepackage{hyperref}
\usepackage{accents}
\usepackage{centernot}
\usepackage{color}
\usepackage{algorithm}
\usepackage[noend]{algpseudocode}
\usepackage[font=scriptsize]{caption}
\usepackage{subcaption} 
\usepackage{mathtools}

\renewcommand{\baselinestretch}{0.95}

\newtheorem{ass}{Assumption}
\newtheorem{thm}{Theorem}

\newtheorem{prop}{Proposition}
\newtheorem{lem}{Lemma}

\title{\LARGE \bf
Error Bounds for Reduced Order Model Predictive Control
}

\author{Joseph Lorenzetti, Marco Pavone	
\thanks{The authors are with the Department of Aeronautics and Astronautics, Stanford University, Stanford CA. Emails: \{jlorenze, pavone\}@stanford.edu.}
\thanks{This work was supported by the Office of Naval Research  (Grant N00014-17-1-2749). Joseph Lorenzetti is supported by the Department of Defense (DoD) through the National Defense Science and Engineering Fellowship (NDSEG) Program.}
}

\begin{document}

\maketitle
\thispagestyle{empty}
\pagestyle{empty}

\begin{abstract}
Model predictive control is a powerful framework for enabling optimal control of constrained systems. However, for systems that are described by high-dimensional state spaces this framework can be too computationally demanding for real-time control. Reduced order model predictive control (ROMPC) frameworks address this issue by leveraging model reduction techniques to compress the state space model used in the online optimal control problem. While this can enable real-time control by decreasing the online computational requirements, these model reductions introduce approximation errors that must be accounted for to guarantee constraint satisfaction and closed-loop stability for the controlled high-dimensional system. 
In this work we propose an offline methodology for efficiently computing error bounds arising from model reduction, and show how they can be used to guarantee constraint satisfaction in a previously proposed ROMPC framework. This work considers linear, discrete, time-invariant systems that are compressed by Petrov-Galerkin projections, and considers output-feedback settings where the system is also subject to bounded disturbances.
\end{abstract}

\section{Introduction}
Some system models, such as those arising from finite approximations to infinite-dimensional systems (e.g. systems described by partial differential equations), may have high-dimensional state spaces that make model-based controller design challenging. This class of systems includes several of practical interest, such as flexible structures \cite{RaoPanEtAl1990}, soft robots \cite{GouryDuriez2018, KatzschmannThieffryEtAl2019}, heating, ventilation, and air conditioning systems \cite{HeGonzalez2016}, coupled fluid/structure systems \cite{AmsallemDeolalikarEtAl2013}, and more. While controllers for these systems could be designed without the use of models, or with low-fidelity hand-engineered models, these may not be suitable for \textit{constrained control} problems or for \textit{optimizing performance}. 

When high-fidelity models of the system are available, constrained optimal control problems are typically approached through receding horizon schemes such as model predictive control (MPC), which optimizes over future behavior by embedding the system model into an optimization problem that is solved online.
However in the case where the high-fidelity models are also high-dimensional, the online computational complexity of the scheme precludes its use for real-time control. Fortunately, principled techniques have been developed to generate reduced order models (ROMs) by compressing the high-dimensional full order models (FOMs) \cite{Antoulas2005}. The resulting dimension reduction can then enable the ROMs to be used for real-time \textit{constrained optimal} control via MPC, while implicitly leveraging the modeling accuracy of the high-fidelity FOM.
However, synthesizing the controller using the ROM could still lead to constraint violations by the controlled FOM, even if the approximation error induced by model reduction is small. Therefore, to \textit{guarantee} constraint satisfaction the approximations errors must be explicitly considered when designing the reduced order model predictive control (ROMPC) scheme. 

\begin{figure}
    \centering
    \includegraphics[width=0.8\columnwidth]{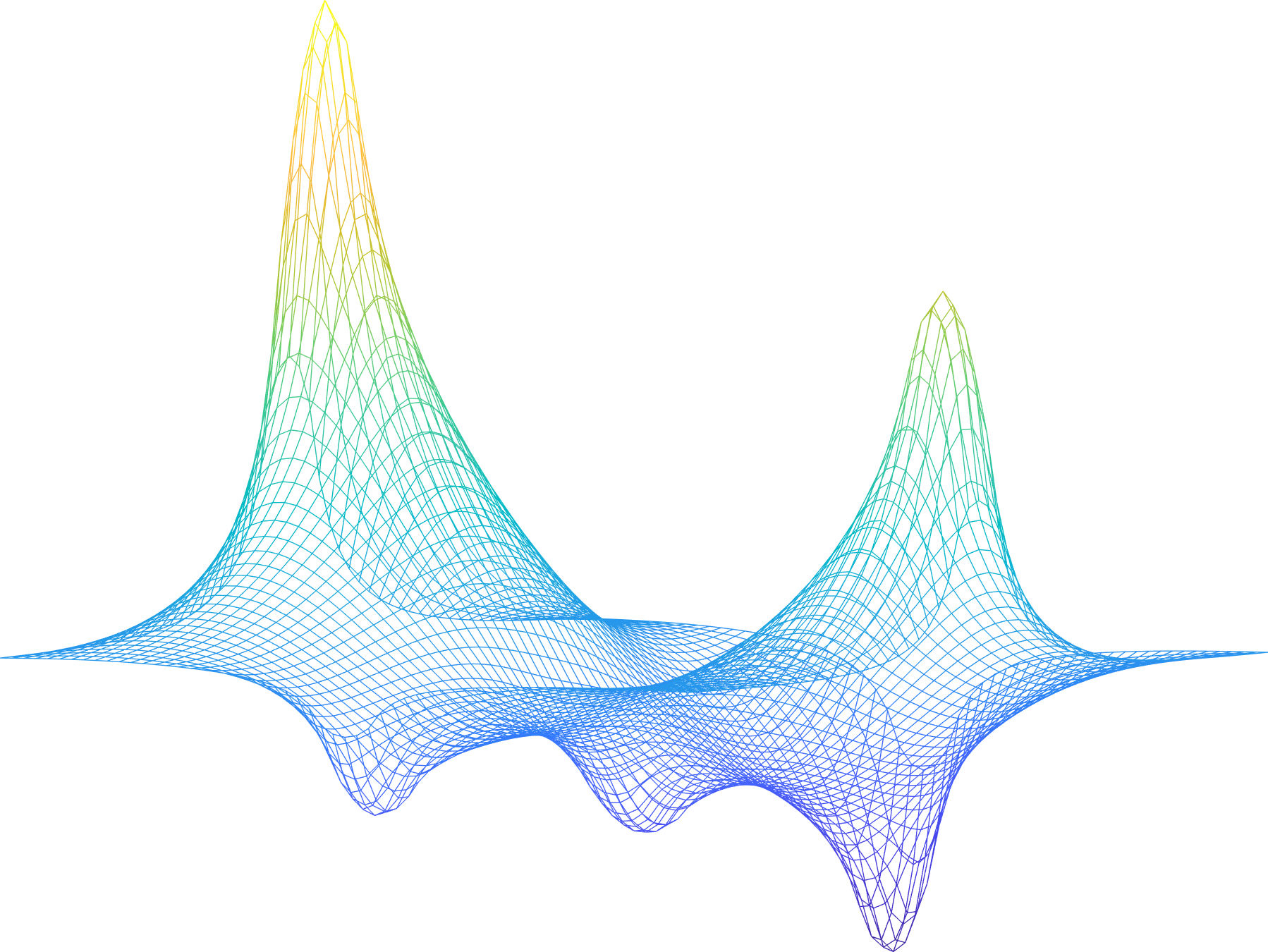}
    \caption{Temperature profile from the 2D heat flow problem discussed in Section \ref{subsec:heatflow}, using a finite element model with $n^f = 3481$ states.}
    \label{fig:tempmesh}
\vspace{-.20in}
\end{figure}

In this work we propose a method for efficiently computing \textit{a priori} bounds on the possible constraint violations, which can then be used during the synthesis of the ROMPC scheme to ensure constraint satisfaction. Specifically, we consider the tube-based output feedback ROMPC scheme proposed by \cite{LorenzettiLandryEtAl2019} which uses ROMs to design both the state estimator and the online optimization problem. In addition to considering possible constraint violation due to the model reduction error we also consider errors arising from state estimation and bounded disturbances.

{\em Related Work:} Previous linear model predictive control schemes that utilize reduced order models vary in how constraints are handled and what kind of performance guarantees can be made. Early works such as \cite{AstridHuismanEtAl2002} and \cite{HovlandGravdahlEtAl2008} do not consider state constraints and also do not provide guarantees on closed-loop stability. Stability guarantees via large semi-definite programs are available for the method from \cite{HovlandLovaasEtAl2008}, however this work only considers ``soft'' state constraints (through cost function penalties). Another method with stability guarantees is \cite{SopasakisBernardiniEtAl2013}, which is the first to also provide constraint satisfaction guarantees. However the constraints are limited to being a function of the reduced order state.

The approach in \cite{LoehningRebleEtAl2014} then considers general state and control constraints, and provides both asymptotic stability and constraint satisfaction guarantees. Although this approach assumes full state feedback, we note that this approach could be extended to output feedback settings. However the disadvantage of this approach comes through the incorporation of a scalar error bounding system into the MPC scheme. This is done such that error bounds do not need to be computed \textit{a priori}, but which leads to increased conservatism that can limit the controller's performance and initial feasible set. 

The use of linear programming to compute \textit{a priori} error bounds is then proposed in \cite{KoegelFindeisen2015b}, which uses a tube-based ROMPC scheme, explicitly considers output feedback, and provides guarantees on stability and constraint satisfaction. This is then extended in \cite{LorenzettiLandryEtAl2019} to yield less conservative error bounds and also consider the problem of setpoint tracking under model reduction error. The main disadvantage with these methods is that the linear programs utilize the FOM dynamics, and therefore do not scale well with problem size. 

In this work we use the ROMPC scheme proposed in \cite{LorenzettiLandryEtAl2019}, but compute error bounds in a novel way that blends some of the ideas from \cite{LoehningRebleEtAl2014}, \cite{KoegelFindeisen2015b}, and \cite{LorenzettiLandryEtAl2019}. This approach is more computationally efficient than the linear programming methods presented in \cite{KoegelFindeisen2015b} and \cite{LorenzettiLandryEtAl2019}, but yields a simpler and less conservative ROMPC scheme than \cite{LoehningRebleEtAl2014}.

{\em Statement of Contributions:}
In this work we propose a novel method for computing \textit{a priori} error bounds that can be used to guarantee robust constraint satisfaction of the controlled high-dimensional system when using a ROMPC scheme. We consider linear, discrete, time-invariant systems where the reduced order model is generated through a Petrov-Galerkin projection of the full order model. Additionally we consider output feedback scenarios where the state estimator is also synthesized using the reduced order model, and where the system is subject to bounded disturbances. This approach improves upon previous methods in efficiency and conservatism of the bounds, and we provide results and comparisons using a small synthetic system described by $n^f = 6$ states as well as a 2D discretized heat flow model described by $n^f = 3481$ states.

{\em Organization:} We begin in Section \ref{sec:prob} by introducing the control problem and the architecture of the ROMPC control scheme. Then in Section \ref{sec:constr_sat} we discuss the robust constraint satisfaction problem that motivates this work. In Sections \ref{sec:bounds} and \ref{sec:G} we propose a novel method for computing error bounds that enable robust constraint satisfaction guarantees to be made. The merits of the approach are then discussed in Section \ref{sec:discussion} before presenting examples in Section \ref{sec:experiments}.

\section{Problem Formulation} \label{sec:prob}
In this section we begin by presenting a mathematical formulation for the control problem and define the system dynamics models and ROMPC architecture.

\subsection{Full Order Model}
Our interest is to control high-dimensional systems, such as those arising from finite approximations to infinite-dimensional systems (e.g. discretized partial differential equations). Specifically we consider cases where the full order system model is expressed in the form
\begin{equation} \label{eq:fom}
\begin{split}
x^f_{k+1} &= A^fx^f_k + B^fu_k + B^f_w w_k, \\
y_k &= C^f x^f_k + v_k, \quad z_k = H^f x^f_k,\\
\end{split}
\end{equation}
where $x^f \in \mathbb{R}^{n^f}$ is the state, $u \in \mathbb{R}^m$ is the control input, $y \in \mathbb{R}^p$ is the measured output, $z \in \mathbb{R}^o$ are performance variables, $w \in \mathbb{R}^{m_w}$ represents disturbances acting on the system dynamics, and $v \in \mathbb{R}^{p}$ is the measurement noise. It is assumed that the disturbances $w$ and $v$ are constrained by
\begin{equation} \label{eq:noise}
w \in \mathcal{W}, \:\: v \in \mathcal{V},
\end{equation}
where $\mathcal{W} \coloneqq \{w \: | \: H_w w \leq b_w \}$ and $\mathcal{V} \coloneqq \{v \: | \: H_v v \leq b_v \}$ are convex polytopes.

Constraints on the performance variables $z$ and control $u$ are also considered, which we assume to be defined by
\begin{equation} \label{eq:constraints}
z \in \mathcal{Z}, \:\: u \in \mathcal{U}, 
\end{equation}
where $\mathcal{Z} \coloneqq \{z\:|\:H_z z \leq b_z \}$ with $H_z \in \mathbb{R}^{n_z \times o}$ and $\mathcal{U} \coloneqq \{u\:|\:H_u u \leq b_u \}$ with $H_u \in \mathbb{R}^{n_u \times m}$ are also convex polytopes. The following assumption is also made regarding the constraints and disturbances, which will be required for the error bound computations presented later.
\begin{ass} \label{ass:compact}
The sets $\mathcal{Z}$, $\mathcal{U}$, $\mathcal{W}$, and $\mathcal{V}$ are compact.
\end{ass}

\subsection{Reduced Order Model}
The reduced order model is defined by
\begin{equation} \label{eq:rom}
\begin{split}
\bar{x}_{k+1} &= A\bar{x}_k + B\bar{u}_k, \\
\bar{y}_k &= C \bar{x}_k, \quad \bar{z}_k = H \bar{x}_k, \\
\end{split}
\end{equation}
where $\bar{x} \in \mathbb{R}^{n}$ is the nominal reduced order state, $\bar{u} \in \mathbb{R}^m$ is the control input to the reduced order system, and $\bar{y} \in \mathbb{R}^p$ and $\bar{z} \in \mathbb{R}^o$ are the nominal measured and performance outputs of the reduced order model, respectively.

In this work we assume the reduced order model was generated via a projection-based method where an oblique Petrov-Galerkin projection $P = V(W^TV)^{-1}W^T$ is used. With this definition, the high-dimensional state can be projected to the reduced order state by $\bar{x} = (W^TV)^{-1}W^Tx^f$ and then can be approximately reconstructed by $x^f \approx V\bar{x}$, and the reduced order model is defined by $A \coloneqq (W^TV)^{-1}W^TA^fV$, $B \coloneqq (W^TV)^{-1}W^TB^f$, $C \coloneqq C^fV$, $H \coloneqq H^fV$. Common model reduction methods such as balanced truncation and proper orthogonal decomposition (POD) can be defined in this way \cite{Antoulas2005}.  We also make the following assumption which is used for controller design and for error bound computations.
\begin{ass} \label{ass:ctrlobsv}
The pair $(A, B)$ is controllable and the pairs $(A,C)$ and $(A,H)$ are observable.
\end{ass}

\subsection{Reduced Order Model Predictive Control}
To control the full order system described by \eqref{eq:fom}, we consider the formulation of the ROMPC control scheme as presented in \cite{LorenzettiLandryEtAl2019}. This scheme consists of a reduced order state estimator, linear feedback controller, and model predictive controller.

\subsubsection{Reduced Order State Estimator}
The reduced order state estimator is given by
\begin{equation} \label{eq:estimator}
\hat{x}_{k+1} = A\hat{x}_k + Bu_k + L(y_k - C\hat{x}_k),
\end{equation}
where $\hat{x} \in \mathbb{R}^n$ is the reduced order state estimate, $u_k$ and $y_k$ are the control and measurement from the \textit{full order system} \eqref{eq:fom}, and $L$ is the estimator gain matrix.

\subsubsection{Reduced Order Controller}
The reduced order control law that is used to control the full order system is given by
\begin{equation} \label{eq:controller}
u_k = \bar{u}_k + K(\hat{x}_k - \bar{x}_k),
\end{equation}
where $K$ is the controller gain matrix. Additionally, the variables $(\bar{x}_k, \bar{u}_k)$ are the state and control values of a \textit{simulated} ROM with dynamics \eqref{eq:rom} that is controlled via a reduced order optimal control problem. This controller form is commonly used in robust control approaches to track a nominal trajectory and reject disturbances. 

\subsubsection{Reduced Order Optimal Control Problem}
The reduced order optimal control problem (OCP) that is used to control the simulated ROM system is given by
\begin{equation} \label{eq:rompc}
\begin{split}
(\mathbf{\bar{x}^*_k}, \mathbf{\bar{u}^*_k}) = \underset{\mathbf{\bar{x}_k}, \mathbf{\bar{u}_k}}{\text{argmin.}} \:\:& \lVert \bar{x}_{k+N|k} \rVert^2_P + \sum_{j=k}^{k+N-1}\lVert \bar{x}_{j|k}\rVert^2_Q + \lVert \bar{u}_{j|k}\rVert^2_R, \\
\text{subject to} \:\: & \bar{x}_{i+1|k} = A\bar{x}_{i|k} + B\bar{u}_{i|k},  \\
& H\bar{x}_{i|k} \in \bar{\mathcal{Z}}, \:\: \bar{u}_{i|k} \in \bar{\mathcal{U}}, \\
& \bar{x}_{k+N|k} \in \bar{\mathcal{X}}_f,\: \bar{x}_{k|k} = \bar{x}_k,
\end{split}
\end{equation}
where $i=k,\dots,k+N-1$, the initial condition constraint is defined by the current state of the simulated ROM $\bar{x}_k$, and the solution $(\mathbf{\bar{x}^*_k}, \mathbf{\bar{u}^*_k}) = (\{\bar{x}^*_i \}_{i=k}^{k+N}, \{\bar{u}^*_i \}_{i=k}^{k+N-1})$ defines an optimal future trajectory for the simulated ROM over the finite horizon $N$. In this optimization problem the positive definite matrices $Q$ and $R$ define the stage cost, and the positive definite matrix $P$ defines the terminal cost. The constraint sets $\bar{\mathcal{Z}}$ and $\bar{\mathcal{U}}$ are tightened versions of the original constraints \eqref{eq:constraints} such that $\bar{\mathcal{Z}} \subseteq \mathcal{Z}$ and $\bar{\mathcal{U}} \subseteq \mathcal{U}$ and the set $\bar{\mathcal{X}}_f$ defines a terminal state constraint. The simulated ROM state $\bar{x}_{k+1}$ is then given by \eqref{eq:rom} with input $\bar{u}_k = \bar{u}^*_k$.

The tightened constraints $\bar{\mathcal{Z}}$ and $\bar{\mathcal{U}}$ are used to ensure robust constraint satisfaction by leveraging the computed error bounds, and are defined in Section \ref{sec:constr_sat}. We assume that the terminal cost matrix $P$ and the terminal set $\bar{\mathcal{X}}_f$ are chosen such that the reduced order OCP \eqref{eq:rompc} defined for the simulated ROM \eqref{eq:rom} is recursively feasible and exponentially stable such that $\bar{x}_k \rightarrow 0$ and $\bar{u}_k \rightarrow 0$. Procedures for designing $P$ and $\bar{\mathcal{X}}_f$ that guarantee this are described in \cite{RawlingsMayneEtAl2017} or \cite{LorenzettiLandryEtAl2019}.

\subsubsection{ROMPC Algorithm} \label{subsubsec:rompc}
To summarize, the ROMPC scheme is defined by the  state estimator \eqref{eq:estimator}, control law \eqref{eq:controller}, and simulated ROM \eqref{eq:rom} controlled by the reduced order OCP \eqref{eq:rompc}. However, before this control scheme can be applied to control the full order system, both the simulated ROM $\bar{x}$ and the state estimator $\hat{x}$ need to be initialized. The simplest approach would be to initialize them at time $k=0$ to be $\hat{x}_0 = \bar{x}_0 = 0$ and then use a separate ``startup'' controller to give the values time to converge. In particular suppose the ROMPC controller took over at time $k_0$, then for $k \in [0, k_0-1]$ the state estimator would be updated according to \eqref{eq:estimator} and the simulated ROM would be updated according to \eqref{eq:rom} with $\bar{u}_k = u_k - K(\hat{x}_k - \bar{x}_k)$.

\section{Robust Constraint Satisfaction} \label{sec:constr_sat}
By design of the ROMPC control scheme, if the reduced order OCP \eqref{eq:rompc} is feasible at time $k_0$ then the simulated ROM satisfies the constraints \eqref{eq:constraints} such that $\bar{z}_k \in \mathcal{Z}$ and $\bar{u}_k \in \mathcal{U}$ for all $k \geq k_0$.
However this is not sufficient to guarantee that the full order system \eqref{eq:fom} satisfies $z_k \in \mathcal{Z}$ and $u_k \in \mathcal{U}$ for all $k \geq k_0$ due to disturbances, estimation error, and model reduction error.

Consider the state reduction error $e \coloneqq x^f - V\bar{x}$ and control error $d = \hat{x} - \bar{x}$. By definition of the constraint sets \eqref{eq:constraints} it can be seen that
\begin{equation*}
\begin{split}
z \in \mathcal{Z} &\iff H_z(\bar{z} + H^f e)  \leq b_z, \\
u \in \mathcal{U} &\iff H_u(\bar{u} + Kd)  \leq b_u.
\end{split}
\end{equation*}
Thus we choose to define the tightened constraint sets $\bar{\mathcal{Z}}$, $\bar{\mathcal{U}}$ used in the reduced order OCP as
\begin{equation} \label{eq:tightconstraints}
\begin{split}
\bar{\mathcal{Z}} \coloneqq \{ \bar{z} \:|\: H_z \bar{z} \leq b_z - \Delta_z \}, \\
\bar{\mathcal{U}} \coloneqq \{ \bar{u} \:|\: H_u \bar{u} \leq b_u - \Delta_u \}, \\
\end{split}
\end{equation}
where $\Delta_z \in \mathbb{R}^{n_z}$ and $\Delta_u \in \mathbb{R}^{n_u}$ are bounds on the errors $\delta_z \coloneqq z - \bar{z} = H^fe$ and $\delta_u \coloneqq  u - \bar{u} = Kd$ such that $H_z \delta_{z,k} \leq \Delta_z$ and $H_u \delta_{u,k} \leq \Delta_u$ for all $k\geq k_0$. In Section \ref{sec:bounds} we describe a procedure for computing $\Delta_z$ and $\Delta_u$ such that these conditions hold.
\begin{lem}[Robust Constraint Satisfaction] \label{lem:robconst}
Suppose that at time $k_0$ the reduced order OCP \eqref{eq:rompc} is feasible and that $H_z\delta_{z,k} \leq \Delta_z$ and $H_u\delta_{u,k} \leq \Delta_u$ for all $k\geq k_0$. Then, under the proposed control scheme the full order system will robustly satisfy the constraints \eqref{eq:constraints} for all $k\geq k_0$.
\end{lem}

\section{Error Bounds} \label{sec:bounds}
In this section we define the error dynamics and propose a novel way for bounding their effect on constraint violations in the control problem.
To gain computational efficiency over the linear programming methods discussed earlier while not introducing additional conservatism we use a combination of linear programming and scalar norm bounds. 

\subsection{Error Dynamics}
Once again considering the state reduction error $e$ and the control error $d$, the joint error state $\epsilon$ is defined by $\epsilon \coloneqq [e^T,\: d^T]^T$ such that the error dynamics are given by
\begin{equation} \label{eq:edynamics}
\epsilon_{k+1} = A_\epsilon \epsilon_k + B_\epsilon r_k + G_\epsilon \omega_k,
\end{equation}
where
\begin{equation*}
A_\epsilon = \begin{bmatrix}
A^f & B^fK \\ LC^f & A + BK - LC
\end{bmatrix},
\end{equation*}
and
\begin{equation*}
B_\epsilon = \begin{bmatrix}
P_\perp A^fV & P_\perp B^f \\ 0 & 0
\end{bmatrix}, \quad G_\epsilon = \begin{bmatrix}
B^f_w & 0 \\ 0 & L
\end{bmatrix},
\end{equation*}
where $P_\perp = I-V(W^TV)^{-1}W^T$, $r = [\bar{x}^T, \: \bar{u}^T]^T$, and $\omega = [w^T,\:v^T]^T$. Note that we can also write these dynamics for any arbitrary times $k_1$ and $k_2$ as
\begin{equation} \label{eq:error_response}
\epsilon_{k_2} = A_\epsilon^{k_2 - k_1} \epsilon_{k_1} + \sum_{j=k_1}^{k_2-1}A_\epsilon^{k_2-1-j}\Big( B_\epsilon r_j + G_\epsilon \omega_j\Big).
\end{equation}

\subsection{Preliminary Computations}
We begin by presenting some preliminary results that will be used later to define the bounds $\Delta_z$ and $\Delta_u$. First, we compute a bounding set $\bar{\mathcal{X}}$ on the reduced order states $\bar{x}$ that is induced by the constraints $\bar{z} \in \mathcal{Z}$, which is useful to bound the input to the error system \eqref{eq:edynamics}. Second, we compute explicit bounds $C_r$ and $C_\omega$ on the weighted norm of the error system inputs, such that $\lVert B_\epsilon r \rVert_G \leq C_r$  and $\lVert G_\epsilon \omega \rVert_G \leq C_\omega$ where $G$ is a positive definite weighting matrix. Finally, we compute parameters $M$ and $\gamma$ that define a bound on the weighted norm of the matrix powers $\lVert A_\epsilon^i \rVert_G \leq M \gamma^i$, used for describing the natural decay of the error system.

\subsubsection{Computing \texorpdfstring{$\bar{\mathcal{X}}$}{Xbar}}
By Assumption \ref{ass:ctrlobsv}, specifically the observability of the pair $(A,H)$, we see that \textit{over time} the constraints $\bar{z}_k \in \mathcal{Z}$ will restrict the admissible reduced order states $\bar{x}_k$. We formalize this notion by computing a set $\bar{\mathcal{X}}$ that bounds the possible realizations of $\bar{x}$.

Specifically, we define $\bar{\mathcal{X}} \coloneqq \{ \bar{x} \:|\: H_{\bar{x}} \bar{x} \leq b_{\bar{x}} \}$ where for simplicity the rows of $H_{\bar{x}}$ are chosen using the standard basis vectors such that $\bar{\mathcal{X}}$ is a hyper-rectangle. Then each element $l$ of $b_{\bar{x}}$ is computed by the linear program
\begin{equation} \label{eq:Xbar}
\begin{split} 
b_{\bar{x},l} = \underset{\bar{x}, \bar{u}}{\text{maximize}} \quad &h_{\bar{x},l}^T \bar{x}_0, \\
\text{subject to} \quad &\bar{x}_{i+1} = A\bar{x}_i + B\bar{u}_i, \\
&\bar{u}_i \in \mathcal{U}, \quad i \in [0,\dots,\bar{i}-1] \\
&H\bar{x}_i \in \mathcal{Z}, \quad i \in [0,\dots,\bar{i}]
\end{split}
\end{equation}
where $h_{\bar{x},l}^T$ is the $l$-th row of $H_{\bar{x}}$ and $\bar{i}\geq n-1$. Note that increasing $\bar{i}$ only adds new constraints and therefore will never cause $b_{\bar{x},l}$ to increase, but can make it decrease. 

From this definition of $\bar{\mathcal{X}}$ we have the following result
\begin{prop} \label{prop:Xbar}
Suppose Assumptions \ref{ass:compact} and \ref{ass:ctrlobsv} hold, and $\bar{u}_k \in \mathcal{U}$ and $H\bar{x}_k \in \mathcal{Z}$ for all $k \geq \underline{k}$. Then, the set $\bar{\mathcal{X}}$ is compact and $\bar{x}_k \in \bar{\mathcal{X}}$ for all $k \geq \underline{k}$.
\end{prop}
\begin{proof}
We begin by showing $\bar{\mathcal{X}}$ is compact. From Assumption \ref{ass:ctrlobsv}, the matrix $\mathcal{O} \coloneqq \begin{bmatrix} H^T & (HA)^T & \dots & (HA^{n-1})^T \end{bmatrix}^T$ is full rank. Since $\bar{x}_i = A^i \bar{x}_0 + \delta_i$ where $\delta_i = \sum_{j=0}^{i-1}A^{i-1-j}B\bar{u}_j$, the constraints $H\bar{x}_i \in \mathcal{Z}$ in \eqref{eq:Xbar} can be written as $HA^i\bar{x}_0 + H\delta_i \in \mathcal{Z}$. 
Therefore, by Assumption \ref{ass:compact} and the constraints in \eqref{eq:Xbar} (with $\bar{i} \geq n-1$), the terms $\delta_i$ are bounded, which implies the vector $\mathcal{O}\bar{x}_0$ is bounded, and since $\mathcal{O}$ is full rank $\bar{x}_0$ is bounded as well. Thus each element $b_{\bar{x},l}$ is bounded and by choice of $H_{\bar{x}}$ the set $\bar{\mathcal{X}}$ is compact. To prove that $\bar{x}_k \in \bar{\mathcal{X}}$ for all $k \geq \underline{k}$ we simply note that by the proposition assumptions the constraints of \eqref{eq:Xbar} are satisfied for all $k \geq \underline{k}$ and therefore $b_{\bar{x}}$ is a valid upper bound on $H_{\bar{x}}\bar{x}$.
\end{proof}

\subsubsection{Computing \texorpdfstring{$(C_r, C_\omega)$}{(Cr,Cw)}}
To compute bounds $C_r$ and $C_\omega$ on the weighted norm of the inputs to the error system \eqref{eq:edynamics}, consider the following optimization problems:
\begin{equation} \label{eq:CrandCw}
\begin{split}
C_r &= \underset{\bar{x}\in \bar{\mathcal{X}}, \bar{u} \in \mathcal{U}}{\text{maximize}} \quad \lVert B_\epsilon [\bar{x}^T, \bar{u}^T]^T \rVert_G, \\
C_\omega &= \underset{w\in \mathcal{W}, v \in \mathcal{V}}{\text{maximize}} \quad \lVert G_\epsilon [w^T, v^T]^T \rVert_G,
\end{split}
\end{equation}
where $G$ is a positive definite weighting matrix. From these definitions we have the following result
\begin{prop} \label{prop:CrandCw}
Suppose Assumptions \ref{ass:compact} and \ref{ass:ctrlobsv} hold, and $\bar{u}_k \in \mathcal{U}$ and $H\bar{x}_k \in \mathcal{Z}$ for all $k \geq \underline{k}$. Then, $\lVert B_\epsilon r_k \rVert_G \leq C_r$ and $\lVert G_\epsilon \omega_k \rVert_G \leq C_\omega$ for all $k \geq \underline{k}$, and $C_r$ and $C_\omega$ are finite.
\end{prop}
\begin{proof}
By Proposition \ref{prop:Xbar} and Assumption \ref{ass:compact} the sets $\bar{\mathcal{X}}$, $\mathcal{U}$, $\mathcal{V}$, and $\mathcal{W}$ are compact, which guarantees $C_r$ and $C_\omega$ are finite. Additionally, by the bounded disturbance assumption it is guaranteed that $\omega_k \in \mathcal{W}\times \mathcal{V}$ for all $k$ and therefore $\lVert G_\epsilon \omega_k \rVert_G \leq C_\omega$ for all $k$. Similarly, by the proposition assumptions $r_k \in \bar{\mathcal{X}}\times \mathcal{U}$ for all $k\geq \underline{k}$ such that $\lVert B_\epsilon r_k \rVert_G \leq C_r$ for all $k \geq \underline{k}$.
\end{proof}
In general the optimization problems \eqref{eq:CrandCw} are NP-complete and therefore challenging to solve. However as the constraints are assumed to be convex polytopes the optimum lies at one of the vertices, and therefore vertex enumeration is one possible solution approach. If vertex enumeration is too challenging there also exist methods for efficiently computing upper bounds on the optimum values \cite{MangasarianShiau1986}, but these may be conservative.

\subsubsection{Computing \texorpdfstring{$(M, \gamma)$}{(M,gamma)}} \label{subsubsec:Mg}
We now propose two methods for computing parameters $M$ and $\gamma$ such that $\lVert A_\epsilon^i \rVert_G \leq M \gamma^i$ for any non-negative integer $i$. First, since $G$ is positive definite we note that by the sub-multiplicative property of the induced matrix norm that $\lVert A_\epsilon^i \rVert_G \leq \lVert A_\epsilon \rVert_G^i$. Therefore if $\gamma \coloneqq \lVert A_\epsilon \rVert_G$ and $M\coloneqq 1$ we have that $\lVert A_\epsilon^i \rVert_G \leq M\gamma^i$.

A second method, assuming $A_\epsilon$ is diagonalizable, is to first compute its eigenvalue decomposition, $A_\epsilon = T D T^{-1}$ where $D$ is diagonal. Noting that $A_\epsilon^i = TD^iT^{-1}$ define $\gamma \coloneqq  \max_j \lvert \lambda_j(A_\epsilon)\rvert$ where $\lambda_j(\cdot)$ denotes the $j$-th eigenvalue of the matrix (i.e. the $j$-th diagonal element of $D$). We then have $\lVert A_\epsilon^i \rVert_G \leq  M\gamma^i$ where $M \coloneqq \lVert G^{1/2} T \rVert_2 \lVert T^{-1}G^{-1/2} \rVert_2$ by using the sub-multiplicative property of the induced matrix norm and the fact that $\lVert X \rVert_G = \lVert G^{1/2} XG^{-1/2} \rVert_2$ for the matrix $X$.

In order for these bounds to be meaningful for analyzing the error system we require that $\gamma < 1$. This requirement will obviously not be satisfied if the error dynamics are unstable and so we make the following assumption.
\begin{ass} \label{ass:Aestable}
The gains $K$ and $L$ are chosen such that $A_\epsilon$ is Schur stable.
\end{ass}
Under Assumption \ref{ass:Aestable} we can see that using the eigenvalue decomposition method is guaranteed to give $\gamma < 1$. On the other hand the first method, where $\gamma \coloneqq \lVert A_\epsilon \rVert_G$, is not guaranteed to yield $\gamma < 1$ without the proper choice of the weighting matrix $G$. A more thorough discussion on the merits of each method, along with techniques for determining $G$, are provided in Section \ref{sec:G}.

\subsection{Error Bounds}
We now present our approach for computing the bounds $\Delta_z$ and $\Delta_u$ used to tighten the constraints in \eqref{eq:tightconstraints}, where the primary goal is to reduce conservatism as much as possible while retaining computational efficiency. To accomplish this we leverage the description of the error dynamics given by \eqref{eq:error_response} and the fact that recent disturbances have more influence on the error than those longer ago. Specifically we split the convolution into two parts and treat the recent terms with more precision by formulating a sequence of linear programs. We begin with the following assumption:
\begin{ass} \label{ass:inf_hor_lp}
The simulated ROM satisfies $\bar{u}_k \in \mathcal{U}$ and $H\bar{x}_k\in\mathcal{Z}$ for all $k \geq \underline{k}$ and $\lVert \epsilon_{\underline{k}} \rVert_G \leq \eta_{\underline{k}}$ where $\underline{k} = k_0 - 2\tau$.
\end{ass}
In Assumption \ref{ass:inf_hor_lp}, $k_0$ is used to represent the time at which the ROMPC scheme takes control of the system, and $\tau$ is a user-defined time horizon parameter. This assumption is therefore an assumption on the behavior of the simulated ROM for a period of time before the ROMPC controller takes over (as discussed in Section \ref{subsubsec:rompc}). It is worth noting that the assumptions on $\bar{u}$ and $\bar{x}$ are easily verified in practice, and high confidence on the satisfaction of the assumption on $\lVert \epsilon_{\underline{k}} \rVert_G$ is possible by choosing $\eta_{\underline{k}}$ conservatively based on domain knowledge of the possible values of $x^f$. In fact, as will be seen in the following sections, the value of $\eta_{\underline{k}}$ can be chosen extremely conservatively without having a major impact on the overall conservativeness of the error bounds.

Using \eqref{eq:error_response} we now divide the error $\epsilon_k$ into two components $\epsilon_{k} = \epsilon_k^{(1)} + \epsilon_k^{(2)}$ where
\begin{equation*}
\begin{split}
\epsilon_k^{(1)} &= A_\epsilon^{k-\underline{k}} \epsilon_{\underline{k}} + \sum_{j=\underline{k}}^{k-\tau-1}A_\epsilon^{k-1-j} \Big(B_\epsilon r_j + G_\epsilon \omega_j \Big), \\
\epsilon_k^{(2)} &= \sum_{j=k - \tau}^{k-1}A_\epsilon^{k-1-j} \Big(B_\epsilon r_j + G_\epsilon \omega_j \Big).
\end{split}
\end{equation*}
The term $\epsilon_k^{(2)}$ represents the contribution from the $\tau$ most recent inputs, and the term $\epsilon_k^{(1)}$ represents the contribution from everything prior. We now study these terms individually, beginning with an analysis of the norm  $\lVert \epsilon_k^{(1)}\rVert_G$.

\subsubsection{Bounding \texorpdfstring{$\epsilon_k^{(1)}$}{(e1)}}
It is important to note that since $\epsilon_k^{(1)}$ represents contributions from \textit{older} inputs to the error system, their influence on $\epsilon_k$ is inherently less significant. Therefore, as previously mentioned, an analysis of the norm of $\epsilon_k^{(1)}$ can be used for computational efficiency without adding significant conservativeness to the total error bound.

Through the triangle inequality and by definition of the induced matrix norms we have
\begin{equation*}
\begin{split}
\lVert \epsilon_k^{(1)} \rVert_G \leq& \lVert A_\epsilon^{k-\underline{k}}\rVert_G \lVert \epsilon_{\underline{k}}\rVert_G + \sum_{j=\underline{k}}^{k-\tau-1}\lVert A_\epsilon^{k-1-j}\rVert_G \lVert B_\epsilon r_j\rVert_G \\
& +\sum_{j=\underline{k}}^{k-\tau-1}\lVert A_\epsilon^{k-1-j}\rVert_G \lVert G_\epsilon \omega_j\rVert_G. \\
\end{split}
\end{equation*}
Utilizing the previously computed bounds defined by $M$, $\gamma$, $C_r$ and $C_\omega$ and with the assumed bound on $\lVert\epsilon_{\underline{k}} \rVert_G$ from Assumption \ref{ass:inf_hor_lp} this reduces to
\begin{equation*}
\begin{split}
\lVert \epsilon_k^{(1)} \rVert_G \leq& M\gamma^{k-\underline{k}}\eta_{\underline{k}} +\frac{M \gamma^\tau (C_r + C_\omega)}{1-\gamma},\\
\end{split}
\end{equation*}
where in the last term we also made a substitution from the fact that $\sum_{j=0}^{k-\tau-1} \gamma^{k-1-j} \leq \frac{\gamma^\tau}{1-\gamma}$ since $\gamma < 1$.
Additionally, since $k - \underline{k} \geq 2\tau$ for all $k \geq k_0$ (as defined in Assumption \ref{ass:inf_hor_lp}) it holds that $\gamma^{k-\underline{k}} \leq \gamma^{2\tau}$ for all $k \geq k_0$.
Therefore, we can define a bound $\Delta^{(1)}$ such that $\lVert \epsilon_k^{(1)} \rVert_G \leq \Delta^{(1)}$ for all $k \geq k_0$ where
\begin{equation} \label{eq:inf_hor_lp_e1}
\begin{split}
\Delta^{(1)} \coloneqq M\gamma^{2\tau} \eta_{\underline{k}} +\frac{M\gamma^\tau(C_r + C_\omega) }{1-\gamma}.
\end{split}
\end{equation}
Notice that because $\gamma < 1$ the impact of the assumed bound $\eta_{\underline{k}}$ (defined in Assumption \ref{ass:inf_hor_lp}) on $\Delta^{(1)}$ can be made negligible by choosing $\tau$ to be sufficiently large. This means that $\eta_{\underline{k}}$ can be chosen to be large (i.e. extremely conservative) such that Assumption \ref{ass:inf_hor_lp} is justifiable.

\subsubsection{Bounding \texorpdfstring{$\epsilon_k^{(2)}$}{(e2)}}
We now study the second component of the error, $\epsilon_k^{(2)}$. Here we note that for constraint tightening \eqref{eq:tightconstraints} we are actually interested in bounding $H_z\delta_{z,k}$ and $H_u\delta_{u,k}$ which can be expressed as $E_z\epsilon_k$ and $E_u\epsilon_k$ with
\begin{equation}
E_z = \begin{bmatrix}
H_z H^f & 0
\end{bmatrix}, \quad E_u = \begin{bmatrix}
 0 & H_u K
\end{bmatrix}.
\end{equation}
We therefore develop a method for computing a bound $\Delta^{(2)}(\theta)$ on an arbitrary inner product such that $\theta^T\epsilon^{(2)}_k \leq \Delta^{(2)}(\theta)$. Specifically this is accomplished by solving the linear program
\begin{equation} \label{eq:inf_hor_lp_e2}
\begin{split} 
\Delta^{(2)}(\theta)= \underset{\bar{x}, \bar{u}, \omega}{\text{max.}} \quad &\theta^T \sum_{j=0}^{\tau-1}A_\epsilon^{\tau-j-1}\big(B_\epsilon r_j + G_\epsilon \omega_j\big), \\
\text{subject to} \quad & \bar{x}_{i+1} = A\bar{x}_i + B\bar{u}_i,\\
&r_i \in \bar{\mathcal{X}} \times \mathcal{U}, \:\: r_i = [\bar{x}_i^T, \bar{u}_i^T]^T, \\
&H\bar{x}_i \in \mathcal{Z}, \quad i \in [-\tau, \dots, \tau-1], \\
&\omega_i \in \mathcal{W} \times \mathcal{V}, \quad i \in [0,\dots, \tau-1].
\end{split}
\end{equation}
As can be seen the objective function comes directly from the definition of $\epsilon_k^{(2)}$, and so this linear program computes the worst case value of $\theta^T\epsilon^{(2)}_k$ under the assumption that the values $r_i$ are unknown over the horizon $i \in [k-\tau, \dots ,k-1]$, but that $H\bar{z}_i \in \mathcal{Z}$ and $\bar{u}_i \in \mathcal{U}$ over the horizon $i \in [k-2\tau,\dots,k-1]$. It is therefore readily apparent from Assumption \ref{ass:inf_hor_lp} that this bound is valid for all $k \geq k_0$.

It is important to note that the linear program \eqref{eq:inf_hor_lp_e2} can be solved very efficiently since the decision variables $\bar{x}$ and $\bar{u}$ are the reduced order state and control, and while the objective function contains the matrices $A_\epsilon$, the objective function itself is not high dimensional and the required matrix products can be efficiently computed.

\subsubsection{Defining \texorpdfstring{$\Delta_z$}{(Dz)} and \texorpdfstring{$\Delta_u$}{(Du)}}
Now that we have shown how to compute bounds on $\epsilon_k^{(1)}$ and $\epsilon_k^{(2)}$ we show how they can be combined to give bounds such that $E_z \epsilon_k = H_z \delta_{z,k} \leq \Delta_z$ and $E_u \epsilon_k = H_u \delta_{u,k} \leq \Delta_u$ that can be used in \eqref{eq:tightconstraints} for constraint tightening.
Considering a row $\theta^T$ of either $E_z$ or $E_u$:
\begin{equation*}
\begin{split}
\theta^T \epsilon_k &= \theta^T(\epsilon^{(1)}_k + \epsilon^{(2)}_k)\\
&\leq \lVert \theta^T G^{-1/2} \rVert_2 \lVert \epsilon^{(1)}_k \rVert_G + \theta^T \epsilon^{(2)}_k.\\
\end{split}
\end{equation*}
Therefore, from the previously computed bounds the $i$-th element of vector $\Delta_z$ and the $j$-th element of vector $\Delta_u$ are defined as
\begin{equation} \label{eq:bounds}
\begin{split}
\Delta_{z,i} &\coloneqq \Delta^{(1)}(e_{z,i}) + \Delta^{(2)}(e_{z,i}), \\
\Delta_{u,j} &\coloneqq \Delta^{(1)}(e_{u,j}) + \Delta^{(2)}(e_{u,j}), \\
\end{split}
\end{equation}
where $i \in [1,\dots,n_z]$, $j \in [1,\dots,n_u]$, $e_{z,i}^T$ and $e_{u,j}^T$ are the rows of $E_z$ and $E_u$ respectively, and
\begin{equation}
\begin{split}
\Delta^{(1)}(\theta) = \lVert \theta^T G^{-1/2} \rVert_2 \Delta^{(1)}.
\end{split}
\end{equation}

\begin{thm}[Robust Constraint Satisfaction]
Suppose that at time $k_0$ the reduced order OCP \eqref{eq:rompc} is feasible where the tightened constraints are defined using the bounds \eqref{eq:bounds}, and that Assumptions \ref{ass:compact}-\ref{ass:inf_hor_lp} hold. Then, under the proposed control scheme the full order system will robustly satisfy the constraints \eqref{eq:constraints} for all $k \geq k_0$.
\end{thm}
\begin{proof}
By design, under Assumptions \ref{ass:compact}-\ref{ass:inf_hor_lp} the bounds $\lVert \epsilon_k^{(1)} \rVert_G \leq \Delta^{(1)}$ and $\theta^T\epsilon_k^{(2)} \leq \Delta^{(2)}(\theta)$ hold for any $\theta$ and for all $k \geq k_0$. Therefore by definition $H_z\delta_{z,k} \leq \Delta_z$ and $H_u\delta_{u,k} \leq \Delta_u$ for all $k \geq k_0$. From Lemma \ref{lem:robconst} we then have the final result.
\end{proof}

\section{Computing Parameters \texorpdfstring{$(M, \gamma, G)$}{(M, gamma, G)}} \label{sec:G}
In Section \ref{subsubsec:Mg} we introduced two methods for computing the parameters $M \geq 1$ and $\gamma < 1$, such that for a Schur stable matrix $P \in \mathbb{R}^{p \times p}$, $\lVert P^i \rVert_G \leq M \gamma^i$ for all integers $i\geq 0$. Here we present more concrete techniques for computing these parameters, including the positive definite matrix $G \in \mathbb{R}^{p \times p}$, and also consider how these choices will impact the bounds defined by \eqref{eq:bounds}.

\subsection{Discrete Lyapunov Method} \label{subsec:lyap}
In this approach the matrix $G$ is computed to ensure $\lVert P \rVert_G < 1$, such that $\gamma \coloneqq \lVert P \rVert_G$ and $M \coloneqq 1$ can be chosen since $\lVert P^i \rVert_G \leq \lVert P \rVert_G^i = \gamma^i$. Intuitively this approach finds a matrix $G$ which corresponds to a transformation under which the system is strictly contracting. This is accomplished by solving the discrete Lyapunov equation
\begin{equation}
P^T G P - \eta^2 G + I = 0,
\end{equation}
for which a unique positive definite solution is guaranteed to exist when $P$ is Schur stable and $\eta \in ( \max_j \lvert \lambda_j(P)\rvert, 1)$ where $\lambda_j(P)$ is the $j$-th eigenvalue of $P$.
Finally, since $\lVert P \rVert_G = \lVert G^{1/2}P G^{-1/2} \rVert_2 = \max_i \sqrt{\lambda_i(\eta^2I - G^{-1})} < \eta$ where $\lambda_i(\cdot)$ denotes the $i$-th eigenvalue of the matrix, we are guaranteed to have $\lVert P \rVert_G < \eta$.


\subsection{Geometric Programming Method} \label{subsec:geoprog}
For this method, assuming the Schur stable matrix $P$ is diagonalizable, we consider the eigenvalue decomposition method presented in Section \ref{subsubsec:Mg} and compute $P = T D T^{-1}$ where $D$ is diagonal. We then choose $\gamma \coloneqq  \max_i \lvert \lambda_i(P)\rvert < 1$ and set $M\coloneqq \lVert G^{1/2} T \rVert_2 \lVert T^{-1}G^{-1/2} \rVert_2$. 
Note that this technique uses the minimum value of $\gamma$ possible, but in general $M\neq 1$ as in the previous methods. The matrix $G$ is then computed to attempt to minimize the multiplicative constant. We therefore propose a general optimization-based definition of $G$ as
\begin{equation} \label{eq:idealGprob}
\begin{split}
G \coloneqq \underset{G > 0}{\text{minimize}} \quad &\prod_l \lVert G^{1/2} X_l \rVert_2 \lVert Y_l G^{-1/2} \rVert_2 , \\
\end{split}    
\end{equation}
where $X_l \in \mathbb{R}^{p \times m_l}$ and $Y_l \in \mathbb{R}^{n_l \times p}$ are problem specific matrices. For example to minimize $M$ you would choose $X = T$ and $Y = T^{-1}$. However this problem is generally intractable to solve and therefore we make some approximations. Specifically, by noting that $\lVert \cdot \rVert_2 \leq \lVert \cdot \rVert_F$, where $ \lVert \cdot \rVert_F$ denotes the Frobenius norm, and by restricting $G$ to be diagonal we can see that
\begin{equation*}
\prod_l \lVert G^{1/2} X_l \rVert_2 \lVert Y_l G^{-1/2} \rVert_2 \leq \prod_l \sqrt{\sum_{i=1}^p g_i a_{l,i}} \sqrt{\sum_{j=1}^p \frac{1}{g_j} b_{l,j}},
\end{equation*}
where $g_i$ is the $i$-th diagonal element of $G$, and $a_{l,i} = \sum_{j=1}^{m_l} (X_l)_{i,j}^2$, and $b_{l,j} = \sum_{i=1}^{n_l} (Y_l)_{i,j}^2$ where $(X_l)_{i,j}$ denotes the $(i,j)$-th element of $X_l$.

We then choose $G$ to minimize this upper bound by formulating the following optimization problem
\begin{equation} \label{eq:geoprog}
\begin{split}
G \coloneqq \underset{g_i > 0, i \in \{1, \dots, p\}}{\text{minimize}} \quad &\prod_l \sum_{i,j=1}^p a_{l,i} b_{l,j} g_i g_j^{-1}. \\
\end{split}    
\end{equation}
By definition $a_{l,i} \geq 0$ and $b_{l,j} \geq 0$ and therefore this problem is a geometric program which can be transformed and solved as a convex optimization problem \cite{BoydVandenberghe2004}. For problems where $p$ is very large this problem can also be solved iteratively using a batch coordinate descent strategy starting with an initial guess for $G$ that is positive definite. For the error bounding method in Section \ref{sec:bounds} we propose the use of $X_1 = T$, $Y_1 = T^{-1}$, $X_2 = \text{diag}(B_\epsilon, \: G_\epsilon)$, and $Y_2 = \text{diag}(E_z, \: E_u)$ to decrease $\Delta^{(1)}$.

\subsection{Comparison of Methods}
For problems where the full order system dimension is low enough that a Lyapunov equation can be solved the Lyapunov method is simple and effective. However for larger problems where the Lyapunov equation cannot be solved efficiently the geometric programming method is a good alternative since the problem can be broken up into smaller parts and solved sequentially.

\section{Discussion} \label{sec:discussion}
The goal of this work is to develop a computationally efficient method for computing error bounds for ROMPC that does not hinder performance by being overly conservative. In this section we discuss the differences between previous ROMPC methods to highlight the advantages of our approach.

\subsection{ROMPC with \textit{A Priori} Error Bounding Methods}
Previous ROMPC schemes that also consider general constraints and use \textit{a priori} error bounds include \cite{KoegelFindeisen2015b}, and \cite{LorenzettiLandryEtAl2019}. In fact these methods are similar to each other as they both compute the error bounds by solving a sequence of linear programs, but through some modifications \cite{LorenzettiLandryEtAl2019} yields less conservative results. The main disadvantage of these methods is that the linear programs include the FOM dynamics as constraints and therefore their complexity scales poorly with FOM size. Additionally, the number of linear programs that need to be solved increases with the dimension of the ROM. Our proposed approach addresses these computational efficiency issues by formulating the linear programs \eqref{eq:inf_hor_lp_e2} with only the ROM dynamics as constraints such that the complexity is greatly reduced. Additionally, the number of linear programs solved is constant with respect to both FOM and ROM dimension.

While conservatism is introduced into our method from the approximations used to compute $\Delta^{(1)}$, by choosing $\tau$ to be large enough the effect of this conservatism will be negligible. Additionally, from the definition of the linear programs used to compute the $\Delta^{(2)}$ terms our approach will have an overall conservatism that is very similar to \cite{LorenzettiLandryEtAl2019}. We demonstrate this for the synthetic system in Section \ref{subsec:synthetic}.

\subsection{ROMPC with \textit{A Posteriori} Error Bounding Methods}
Of the previously proposed ROMPC approaches, the only one that does not utilize \textit{a priori} bounds is \cite{LoehningRebleEtAl2014}. This approach instead incorporates a scalar error bounding system
\begin{equation*}
\Delta_{k+1} = \gamma \Delta_k + M \lVert B_\epsilon r_k \rVert_G.
\end{equation*} 
into the reduced order OCP and obtains stability guarantees by designing the terminal set
\begin{equation*}
\bar{\mathcal{X}}_f \coloneqq \{ (\bar{x}, \:\Delta) \:|\: \bar{x}^TP\bar{x} \leq \gamma_1, \:\: 0 \leq \Delta \leq \gamma_2 \},
\end{equation*}
to be control invariant with respect to the \textit{coupled} ROM/error system dynamics, where $P$, $\gamma_1$, and $\gamma_2$ are computed using the procedure in \cite[Lemma~1]{LoehningRebleEtAl2014}.

Computationally this approach is simpler than our proposed approach and the other \textit{a priori} methods, but suffers from being overly conservative even though a worst-case analysis is not required.
This conservatism is inherent to the use of the scalar bound on the norm of the error, which can be a poor estimate that over-compensates for the actual error dynamics. In contrast, the use of linear programs in our approach leads to much tighter bounds.
We again demonstrate this using the synthetic system in Section \ref{subsec:synthetic}. Note that \cite{LoehningRebleEtAl2014} considers continuous time systems and assumes full state feedback with no disturbances and is therefore challenging to directly compare to our method. Instead we compare against the extension of the method to the discrete time, output feedback setting which is relatively straightforward.

\section{Experiments} \label{sec:experiments}
 
\subsection{Synthetic System} \label{subsec:synthetic}
This system, with dimension $n^f = 6$, is adapted from \cite{HovlandLovaasEtAl2008} and is the same system used in \cite{LorenzettiLandryEtAl2019}. Specifically we define $A^f$ and $B^f$ identically to \cite{HovlandLovaasEtAl2008} and define $C_f = \begin{bmatrix}1.29 & 0.24 & 0_{1 \times 4} \end{bmatrix}$, $H^f = \begin{bmatrix}I_{2 \times 2} & 0_{2 \times 4} \end{bmatrix}$, and $B^f_w = I$. The performance and control constraints are given by $\mathcal{Z} = \{z \:|\: \lVert z \rVert_\infty \leq 50 \}$ and $\mathcal{U} = \{u \:|\: \lVert u \rVert_\infty \leq 20 \}$, and the disturbances are bounded with $\mathcal{W} = \{w \:|\: \lVert w \rVert_\infty \leq 0.05 \}$ and $\mathcal{V} = \{v \:|\: \lVert v \rVert_\infty \leq 0.01 \}$. The reduced order model was computed using balanced truncation with $n=4$. For the control problem we use a time horizon of $N=20$, define the cost with $Q^f = 10I$ and $R^f = I$, choose $Q = V^TQ^fV$ and $R=R^f$, and the controller gains $K$ and $L$ are the linear quadratic regulator gains computed using the ROM dynamics. Finally, the error bounds $\Delta_z$ and $\Delta_u$ were computed with $\tau = 100$ where $G$, $M$, and $\gamma$ were computed using the Lyapunov method from Section \ref{subsec:lyap}, and $C_r$ and $C_\omega$ were computed using vertex enumeration.

In Figure \ref{fig:synthetic_z} we show an example of the closed-loop behavior when the full order system starts from a steady state where $z_2 = 25$, where we also compare the ROMPC scheme to the \textit{full order} robust output feedback MPC controller described in \cite{LorenzettiPavone2020}. It can be seen that under the ROMPC controller the full order system does not approach the constraint boundary as closely as in the full order MPC case (and incurs approximately $6\%$ more cost) due to the error bounds that account for model reduction error.

\begin{figure}[ht]
    \centering
    \includegraphics[width=0.9\columnwidth]{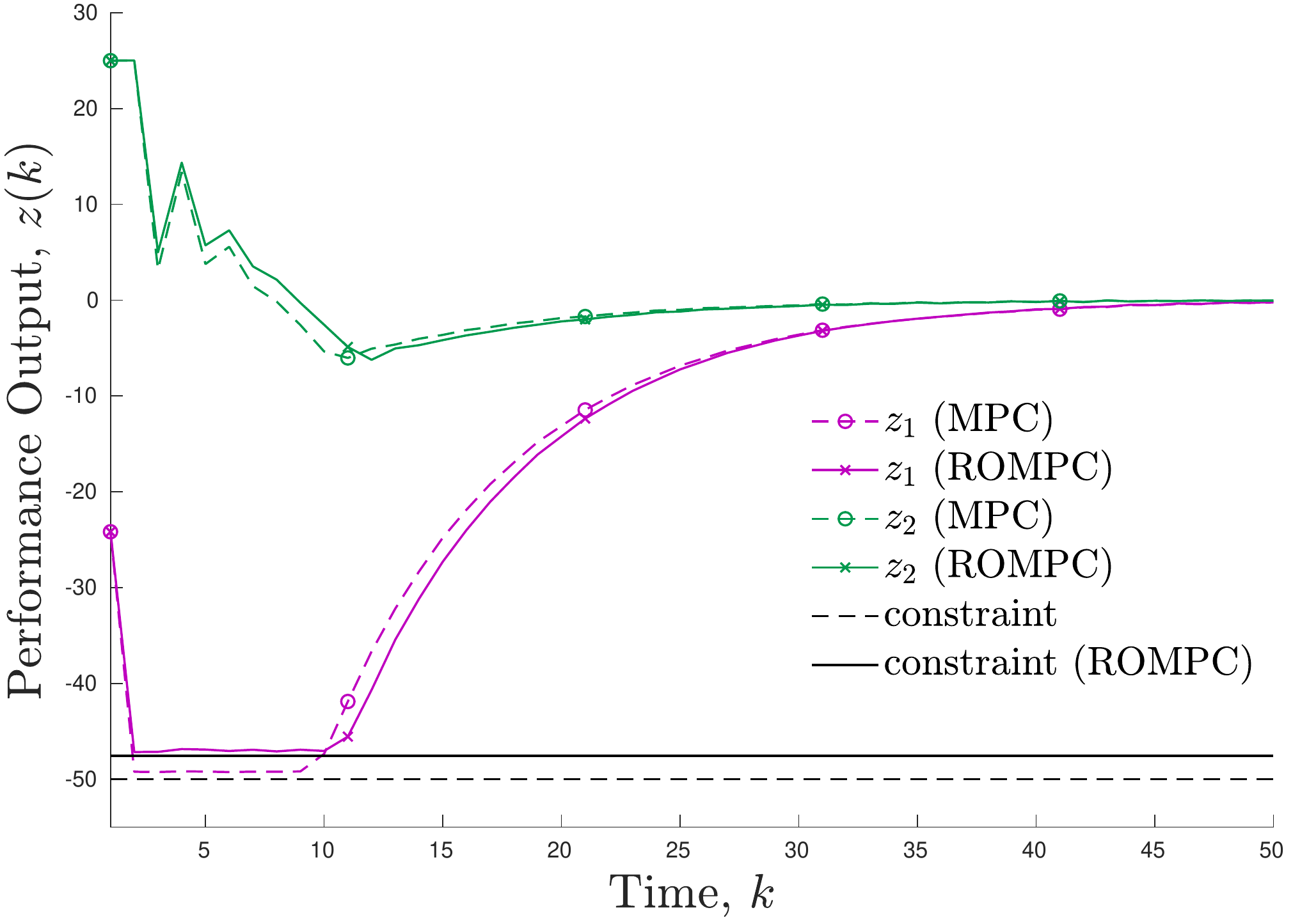}
    \caption{Closed-loop behavior of the synthetic system described in Section \ref{subsec:synthetic} using both a ROMPC controller as well as a full order MPC controller for comparison.}
\label{fig:synthetic_z}
\vspace{-.10in}
\end{figure}

For a comparison to other ROMPC approaches, we first use Algorithm 1 from \cite{LorenzettiLandryEtAl2019} to compute the tightened constraints. In this case the amount of constraint tightening on $z$ changed by less than $1\%$ while the constraints on $u$ were tightened approximately $27\%$ less with our proposed method, showing that in this case our bounds are actually less conservative. 

To compare against the method from \cite{LoehningRebleEtAl2014} we first compute a polytopic outer approximation to the ellipsoidal terminal set $\bar{\mathcal{X}}_f$ used in \cite{LoehningRebleEtAl2014} and note that on average the terminal set from our proposed method is approximately $15$ times larger in each direction. This suggests that the ROMPC scheme from \cite{LoehningRebleEtAl2014} may have a more restricted feasible set for a given horizon $N$. Next we again consider the case where the system starts at steady state corresponding to the variable $z_2$ tracking a setpoint. We find that whenever the setpoint satisfies $\lvert z_2 \rvert \geq 8$, every constraint is tightened more when using \cite{LoehningRebleEtAl2014} than when using our method.
Additionally, when using the method from \cite{LoehningRebleEtAl2014} the ROM initial state corresponding to the $z_2$ setpoint is only feasible with respect to the tightened constraints (i.e. the initial condition is in the ROMPC feasible set) if $\lvert z_2 \rvert \leq 20$. On the other hand with our approach the initial condition is feasible whenever $\lvert z_2 \rvert \leq 27$, which is a direct result of having less conservative constraint tightening. For comparison the initial state is feasible with respect to the \textit{original} constraints for $\lvert z_2 \rvert \leq 28$. For this example these results clearly indicate that our approach is less conservative even though we consider \textit{a priori} error bounds.

\subsection{2D Heat Flow Model} \label{subsec:heatflow}
We also briefly introduce another example to show that our approach can be easily applied to a problem arising from finite element modeling. This system is described by a linear 2D heat flow model adapted from the HF2D9 model in \cite{Leibfritz2006} and discretized in time. This system has $n^f = 3841$ states, $m = 5$ inputs, $p = 5$ measurements, and $o = 4$ performance variables (that represent average temperatures across different small regions) and has an unstable mode. The performance and control constraints are given by $\mathcal{Z} = \{z \:|\: \lVert z \rVert_\infty \leq 2 \}$ and $\mathcal{U} = \{u \:|\: \lVert u \rVert_\infty \leq 100 \}$ and no disturbances are considered. The reduced order model was computed using balanced truncation with $n=21$ after performing a stable/unstable decomposition of the system. For the control problem we use a time horizon of $N=30$ and define the controllers with the linear quadratic regulator gains using the ROM dynamics. The error bounds were computed with $\tau = 2500$ and where $G$, $M$, and $\gamma$ were computed using the geometric programming method from Section \ref{subsec:geoprog} with a batch coordinate descent scheme. The values of $C_r$ and $C_\omega$ were computed via vertex enumeration.

In Figure \ref{fig:heatflow} we show simulation results where the full order system begins at equilibrium at the origin and is then controlled to track a setpoint (using setpoint tracking results from \cite{LorenzettiLandryEtAl2019}). The steady state temperature profile that corresponds to this tracking problem is shown in Figure \ref{fig:tempmesh}. It can be seen that in this particular simulation the control constraints are active at the beginning, where a small gap is present due to the use of tightened constraints resulting from the computed error bounds.

\begin{figure}
    \centering
    \begin{subfigure}[t]{.45\textwidth}
        \includegraphics[width=\textwidth]{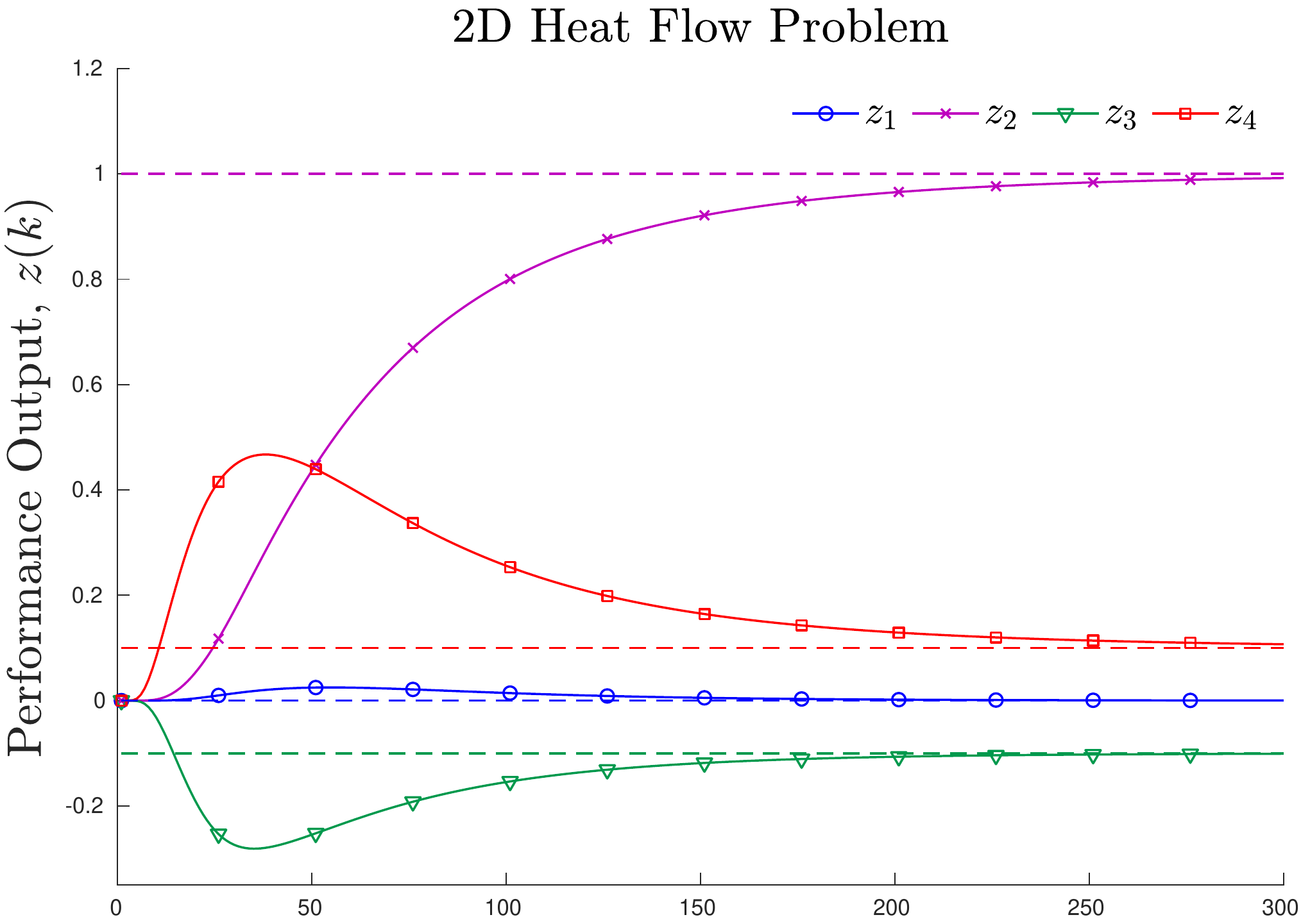}
    \end{subfigure}
    \begin{subfigure}[t]{0.45\textwidth}
        \includegraphics[width=\textwidth]{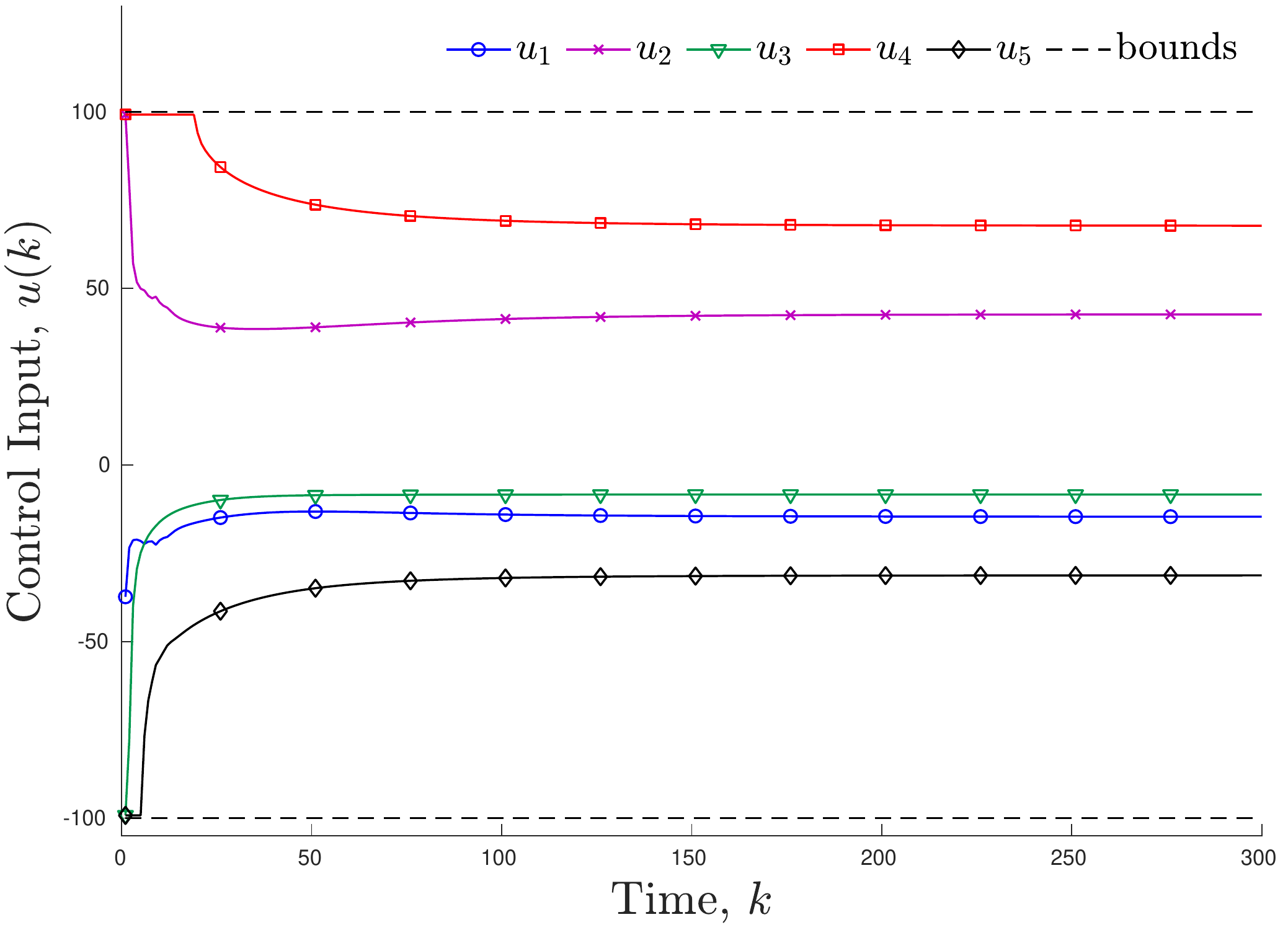}
    \end{subfigure}
    \caption{Simulation of the 2D heat flow setpoint tracking problem discussed in Section \ref{subsec:heatflow}. It can be seen that the control constraints are satisfied and the performance variable targets are tracked.}
\label{fig:heatflow}
\vspace{-.20in}
\end{figure}

\section{Conclusion} \label{sec:conclusion}
In this work we propose a novel method for computing bounds on the error induced by using reduced order models for control, which can be used to design a ROMPC scheme that guarantees constraint satisfaction for the full order system even in the presence of model reduction error, state estimation error, and bounded disturbances. The advantages of the proposed approach include both computational efficiency and reduction in conservatism that leads to better controller performance. The approach was demonstrated on a small synthetic system as well as on a high-dimensional 2D heat flow problem modeled using finite elements.

{\em Future Work:} A consideration that should be addressed in future work is how to compute the controller gains $K$ and $L$ such that Assumption \ref{ass:Aestable} is guaranteed to be satisfied, and to understand how their choice can decrease the error in the control scheme. Additionally, it would be interesting to explore conditions for which controllability and observability of the reduced order system are guaranteed.

\bibliographystyle{IEEEtran}
\bibliography{./main,./ASL_papers}

\end{document}